\documentclass[12pt]{article}
\usepackage{graphicx}
\usepackage{amssymb}
\usepackage{amsfonts}
\usepackage{amsmath}
\usepackage{graphicx}
\usepackage{enumerate}
\usepackage[usenames,dvipsnames]{color} 
\usepackage[right,pagewise,displaymath, mathlines]{lineno}
\usepackage{epstopdf}
\usepackage{subfigure}

\numberwithin{equation}{section}
\numberwithin{figure}{section}
\numberwithin{table}{section}

\newtheorem{theorem}{Theorem}[section]

\newtheorem{proposition}{Proposition}[section]

\newtheorem{example}{Example}[section]

\newtheorem{note}{Note}[section]

\newenvironment{proof}[1][Proof]{\medskip \noindent\textit{#1.} }{\ \rule{0.5em}{0.5em} \medskip}

\begin{document}
\bibliographystyle{cj}

\vspace*{0mm}

\begin{center}
{\Large \bf Measuring the lack of monotonicity in functions}

\vspace*{5mm}

{\large  Danang Teguh Qoyyimi$^{\,a,b}$ and Ri\v cardas Zitikis$^{\,b,*}$}

\bigskip

$^{\,a}$\textit{Department of Mathematics, Gadjah Mada University, Yogyakarta 55281, Indonesia}

\medskip

$^{\,b}$\textit{Department of Statistical and Actuarial Sciences,
University of Western Ontario, London, Ontario N6A 5B7, Canada}

\end{center}

\bigskip

\begin{abstract}

Problems in econometrics, insurance, reliability engineering, and statistics quite often rely on the assumption that certain functions are non-decreasing. To  satisfy this requirement, researchers frequently model the underlying phenomena using parametric and semi-parametric families of functions, thus effectively specifying the required shapes of the functions. To tackle these problems in a non-parametric way, in this paper we suggest indices for measuring the lack of monotonicity in functions. We investigate properties of the indices and also offer a convenient computational technique for practical use.

\bigskip

\noindent
\textit{JEL Classification:}
\begin{quote}
C02 - Mathematical Methods\\
C44 - Statistical Decision Theory; Operations Research \\
C51 - Model Construction and Estimation\\
D81 - Criteria for Decision-Making under Risk and Uncertainty
\end{quote}

\medskip

\end{abstract}

\vfill

\noindent
$^{*}${\small Corresponding author:
tel: +1 519 432 7370;
fax: +1 519 661 3813;
e-mail: \texttt{zitikis@stats.uwo.ca}}

\section{Introduction}
\label{section-1}

In a number of problems such as developing statistical tests, assessing insurance and financial risks, dealing with demand and production functions in economics, modeling mortality and longevity of populations, researchers often face the need to know  whether certain functions are monotonic (e.g., non-decreasing) or not, and if not, then they wish to assess their degree of non-monotonicity. Due to this reason, in this paper we suggest and explore several indices for measuring the lack of non-decreasingness in functions.

While determining monotonicity can be a standard, though perhaps quite difficult,  exercise of checking the sign of the first derivative over the region of interest, assessing the lack of monotonicity in non-monotonic functions has gotten much less attention in the literature (e.g., Davydov and Zitikis, 2005). To illustrate problems where monotonicity, or lack of it, matters, we next present four examples.

\begin{example}\rm\label{exa-1}
Monotone likelihood ratio (MLR) families play important roles in areas of statistics such as constructing uniformly powerful hypothesis tests, confidence bounds and regions. In short, a family of absolutely continuous cumulative distribution functions (cdf's) $\{F_{\theta}: \theta \in \Theta \subseteq \mathbf{R}\}$ is MLR if for every $\theta_1<\theta_2$, the two cdf's $F_{\theta_1}$ and $F_{\theta_2}$ are distinct and the ratio $f_{\theta_1}(\mathbf{x})/f_{\theta_2}(\mathbf{x})$ of the corresponding densities is an increasing function of a statistic $T(\mathbf{x})\in \mathbf{R}$, where $\mathbf{x}=(x_1,\dots , x_n)$ is a generic $n$-dimensional observation. For more details on the MLR families and their uses in statistics, we refer to, e.g., Chapter 4 of Bickel and Doksum (2001).
\end{example}

\begin{example}\rm\label{exa-2}
The presence of a deductible $d\ge 0$ often changes the profile of insurance losses (e.g., Brazauskas et al., 2009). Because of this and other reasons, given two losses $X$ and $Y$, which may not be observable, decision makers wish to determine whether the observable losses $X_d=[X\mid X>d]$ and $Y_d=[Y\mid Y>d]$ are stochastically (ST) ordered, say $X_d \le_{\textrm{ST}} Y_d$ for every $d\ge 0$.  Denuit et al (2005) show on p.~124 that this ordering is equivalent to determining whether the ratio $S_Y(x)/S_X(x)$ is a non-decreasing function in $x$, where $S_X$ and $S_Y$ are the survival functions of $X$ and $Y$, respectively. We conclude this example by noting that this ordering is known in the literature (cf., e.g., Denuit et al, 2005) as the hazard rate (HR) ordering, and is succinctly denoted by $X \le_{\textrm{HR}} Y$.
\end{example}

\begin{example}\rm\label{exa-3}
More generally than in the previous example, one may wish to determine whether for every deductible $d\ge 0$ and every policy limit $L>d$, the observable insurance losses $X_{d,L}=[X\mid d\le X \le L ]$ and $Y_{d,L}=[Y\mid d\le Y \le L ]$ are stochastically ordered, say, $X_{d,L} \le_{\textrm{ST}} X_{d,L} $. We find on pages 127--128 in Denuit et al (2005) that this problem is equivalent to determining whether the ratio $f_Y(x)/f_X(x)$ is a non-decreasing function in $x$ over the union of the supports of $X$ and $Y$, where $f_X$ and $f_Y$ are the density functions of $X$ and $Y$, respectively. This ordering is known in the literature (cf., e.g., Denuit et al, 2005) as the likelihood ratio (LR) ordering and is succinctly denoted by $X \le_{\textrm{LR}} Y$. For further details on various stochastic orderings and their manifold  applications, we refer to Levy (2006), Shaked and Shanthikumar (2006), Li and Li (2013).
\end{example}

\begin{example}\rm\label{exa-4}
Let $\mathcal{X}_+$ denote the set of all non-negative random variables $X$ representing insurance losses. The premium calculation principle (pcp) is a functional $\pi : \mathcal{X}_{+} \to [0 , \infty ] $. Furman and Zitikis (2008a, 2009) have specialized this general premium to the weighted pcp $\pi_w$ defined by the equation $\pi_w[X]=\mathbf{E}[Xw(X)] / \mathbf{E}[w(X)]$, where $w: [0 , \infty ) \to [0 , \infty )$ is a weight function specified by the decision maker, or implied by certain axioms. The functional $\pi_w: \mathcal{X}_{+} \to [0 , \infty ]$ satisfies the non-negative loading property whenever the weight function $w$ is non-decreasing (cf.\, Lehmann, 1966). This is one of the very basic properties that insurance premiums need to satisfy. For further information on this topic, we refer to Sendov et al (2011). For a concise overview of pcp's, we refer to, e.g., Young (2004). For detailed results and their proofs, we refer to, e.g., Denuit et al  (2005).
\end{example}

We next briefly present a few more topics and related references where monotonicity, or lack of it, of certain functions plays an important role:
\begin{itemize}
  \item
  Growth curves (cf., e.g., Bebbington et al, 2009; Chernozhukov et al, 2009; Panik, 2014).
  \item
  Mortality curves  (cf., e.g., Gavrilov and Gavrilova, 1991; Bebbington et al, 2011).
  \item
  Positive regression dependence and risk sharing (cf., e.g., Lehmann,  1966; Barlow and Proschan, 1974; Bebbington et al, 2007; Dana and Scarsini, 2007).
  \item
  Portfolio construction, capital allocations, and comonotonicity (cf., e.g., Dhaene et al, 2002a, 2002b; Dhaene et al, 2006; Furman and Zitikis, 2008b).
  \item
  Decision theory and stochastic ordering (cf., e.g., Denuit et al, 2005; Levy, 2006; Shaked and Shanthikumar, 2006; Egozcue et al, 2013).
  \item
  Engineering reliability and risks (cf., e.g., Barlow and Proschan, 1974;  Lai and Xie, 2006; Singpurwalla, 2006; Bebbington et al, 2008; Li and Li, 2013).
\end{itemize}

One unifying feature of these diverse works is that they impose monotonicity requirements on certain functions, which are generally unknown, and thus researchers seek for statistical models and data for determining their shapes. To illustrate the point, we recall, for example, the work of Bebbington et al (2011) who specifically set out to determine whether mortality continues to increase or starts to decelerate after a certain species related late-life age. This is known in the literature as the late-life mortality deceleration phenomenon. Hence, we can rephrase the phenomenon as a question: is the mortality function always increasing? Naturally, we do not elaborate on this topic any further in this paper, referring the interested reader to Bebbington et al (2011), Bebbington et al (2014), and references therein.

To verify the monotonicity of functions such as those noted in the above examples, researchers quite often assume that the functions belong to some parametric or semiparametric families. One may not, however, be comfortable with this element of subjectivity and thus prefers to rely solely on data to make a judgement. Under these circumstances, verifying monotonicity becomes a non-parametric problem, whose solution asks for an index that, for example, takes on the value $0$ when the function under consideration is non-decreasing and on positive values otherwise. In the following sections we shall introduce and discuss two such indices; both of them are useful, but due to different reasons.

\section{An index of non-decreasingness and its properties}
\label{index-i}

Perhaps the most obvious definition of an index of non-decreasingness is based on the notion of non-decreasing rearrangement, which, for a function $h:[0,1]\to \mathbf{R}$, is defined by
\[
I_{h}(t)=\inf \{x\in \mathbf{R}: G_{h}(x) \ge t\}  \quad \textrm{for all} \quad t \in [0,1] ,
\]
where
\[
G_{h}(x)= \lambda \{ s \in [0,1]: h(s) \leq x \}  \quad \textrm{for all} \quad
x\in \mathbf{R},
\]
with $\lambda$ denoting the Lebesque measure. Hence, any distance between the original function $h$ and its non-decreasing rearrangement $I_{h}$ can serve an index of non-decreasingness. Of course, there are many distances in function spaces, and thus many indices, but we shall concentrate here on the $L_1$-distance due to its attractive geometric interpretation and other properties. Thorough the paper, we assume that $h$ is integrable on its domain of definition.

\begin{note}\label{note21}\rm
The function $I_{h}$ is known in the literature as the generalized inverse of the function $G_{h}$, and is thus frequently denoted by $G_{h}^{-1}$. Throughout this paper, however, we prefer using the notation $I_{h}$ to emphasize the fact that this is a weakly increasing, that is, non-decreasing function. In probability and statistics, researchers would call $I_{h}$ the quantile function of the `random variable' $h$. In the literature on function theory and functional analysis (cf., e.g. Hardy et al, 1952; Denneberg, 1994; Korenovskii, 2007; and references therein) the function $I_h$ is usually called the (non-decreasing) equimeasurable rearrangement of $h$.
\end{note}

We are now in the position to give a rigorous definition of the earlier noted $L_1$-based index of non-decreasingness, which is
\[
\mathcal{I}_{h} = \int_0^1 \big |h(t)-I_{h}(t) \big |dt .
\]
The index $\mathcal{I}_{h} $ takes on the value $0$ if and only if the function $h$ is non-decreasing. The proof of this fact is based on the well-known property (cf., e.g., Proposition \ref{equiv-0} in Appendix \ref{appendix} below) that $h$ is non-decreasing if and only if the equation $I_{h}(t)=h(t)$ holds for $\lambda $-almost all $t \in [0,1]$.

It is instructive to mention here that the notion of monotone rearrangement has been very successfully used in quite a number of areas:
\begin{itemize}
\item
Efficient insurance contracts (e.g., Carlier and Dana, 2005;  Dana and Scarsini, 2007).
\item
Rank-dependent utility theory  (Quiggin, 1982, 1993; also Carlier and Dana, 2003, 2008, 2011).
\item
Continuous-time portfolio selection (e.g., He and Zhou, 2011; Jin and Zhou, 2008).
\item
Statistical applications such as performance improvement of estimators (e.g., Chernozhukov et al, 2009, 2010) and optimization problems (e.g., R\"{u}schendorf, 1983).
\item
Stochastic processes and probability theory  (e.g., Egorov, 1990; Zhukova, 1994, 1998; Thilly, 1999).
\end{itemize}
These are just a few illustrative topics and references, but they lead us into the vast literature on monotone rearrangements and their manifold uses.

The following probabilistic interpretation of the basic quantities involved in our research will play a pivotal role, especially when devising simple proofs of a number of results. We note at the outset that the interpretation is well known and appears frequently in the literature (cf., e.g., Denneberg, 1994; Carlier and Dana, 2005)

\begin{note}[Probabilistic interpretation]\label{note22}\rm
The interval $[0,1]$ can be viewed as a sample space, usually denoted by $\Omega $ in probability and statistics. Furthermore, the Lebesgue measure $\lambda $ can be viewed as a probability measure, usually denoted by $\mathbf{P}$, which is defined on the $\sigma $-algebra of all Borel subsets of $\Omega = [0,1]$. Hence, the function $h:[0,1]\to \mathbf{R}$ can be viewed as a random variable, usually denoted by $X:\Omega \to \mathbf{R}$ in probability and statistics. Under these notational agreements, the function $G_{h}$ can be viewed as the cdf $F_X$ of $X$, and, in turn, the function $I_{h}$ can be viewed as the quantile function $F_X^{-1}$ of $X$.
\end{note}

To illustrate how this probabilistic point of view works, we recall the well-known equation
\begin{equation}
\label{lemma.sameint}
\int_0^1 I_{h}(t)dt=\int_0^1 h(t)dt,
\end{equation}
which we shall later use in proofs. The validity of equation (\ref{lemma.sameint}) can easily be established as follows. We start with the equation $\int_0^1 I_{h}(t)dt=\int_0^1 F_X^{-1}(t)dt$. Then we recall that the mean $\mathbf{E}[X]$ of $X$ can be written as $\int_0^1 F_X^{-1}(t)dt$. Hence,  $\int_0^1 I_{h}(t)dt=\mathbf{E}[X]$. Furthermore, appealing to the probabilistic interpretation one more time, we have $\mathbf{E}[X]=\int_0^1 h(t)dt $, which establishes equation (\ref{lemma.sameint}). Of course, from the purely mathematical point of view, equation (\ref{lemma.sameint}) follows from the fact that $h$ and $I_{h}$ are equimeasurable functions and thus their integrals coincide. In summary, we have demonstrated that equation (\ref{lemma.sameint}) holds.

We conclude this section with a few additional properties of the index $\mathcal{I}_{h}$ which will lead us naturally to the next section. First, as one would intuitively expect, any index of non-decreasingness should not change if the function  $h: [0,1] \to \mathbf{R}$ is lifted up or down by any constant $d\in \mathbf{R}$. This is indeed the case, as the equation
\begin{equation}\label{prop-1u}
\mathcal{I}_{h+d}=\mathcal{I}_{h}
\end{equation}
follows easily upon checking that, for every constant $d\in \mathbf{R}$, the equation  $I_{h+d}(t)=I_{h}(t)+d$ holds for every $ t\in [0,1]$. Finally, the multiplication of the function $h$ by any non-negative constant $c\ge 0$ (so as not to change the direction of monotonicity) should only change the index by as much as it changes the slope of the function. Indeed, we have the equation
\begin{equation}\label{prop-2}
\mathcal{I}_{ch}=c\mathcal{I}_{h}
\end{equation}
that follows easily upon checking that, for every constant $c\ge 0$, the equation $I_{ch}(t)=cI_{h}(t) $ holds for every $ t\in [0,1]$.

\section{Comonotonically additive index of non-decreasingness}
\label{section-3}

It is instructive to view equation (\ref{prop-1u}) as the additivity property
\begin{equation}\label{prop-1uu}
\mathcal{I}_{h+g_0}=\mathcal{I}_{h}+\mathcal{I}_{g_0},
\end{equation}
where $g_0$ is the constant function defined by $g_0(t)=d$ for all $t\in [0,1]$, with $d\in \mathbf{R}$ being a constant. Indeed, $\mathcal{I}_{g_0}=0$, and thus we conclude that equations (\ref{prop-1u}) and (\ref{prop-1uu}) are equivalent.

Note that the functions $h$ and $g_0$ are commonotonic irrespectively of the value of $d$. This fact follows immediately from the definition of comonotonicity (cf. Schmeidler, 1986): \textit{Two functions $h$ and $g$ are comonotonic if and only if there are no $t_1$ and $t_2$ such that  $h(t_1)<h(t_2)$ and $g(t_1)>g(t_2)$.} This is a well-known notion, extensively utilized in many areas, perhaps most notably in economics and insurance. For further details and references on the topic, we refer to Denneberg (1994), Dhaene et al (2002a,b), Dhaene et al (2006), and references therein.

Coming now back to equation (\ref{prop-1uu}), a natural question is whether the equation still holds if the constant function $g_0$ is replaced by any other function $g$ that is comonotonic with $h$. For this, we first recall the fact (cf. Corollary 4.6 in Denneberg, 1994) that, for every pair of comonotonic functions $h$ and $g$,
\begin{equation}\label{prop-1v2}
I_{h+g}(t)=I_{h}(t)+I_{g}(t) \quad \textrm{for every} \quad t\in [0,1].
\end{equation}
Unfortunately, the index $\mathcal{I}_{h}$ is based on the non-linear functional $\Delta \mapsto \int_0^1 |\Delta (t)|dt$, and we can thus at most have the subaddivity property:
\begin{equation}\label{prop-1uu2}
\mathcal{I}_{h+g}\le \mathcal{I}_{h}+\mathcal{I}_{g} .
\end{equation}
(The lack of additivity would, of course, still be the case even if we replaced the $L_1$-type functional by any other $L_p$-type functional.) Hence, we need a linear functional.

Note that by simply dropping the absolute values from the functional $\Delta \mapsto \int_0^1 |\Delta (t)|dt$ would not lead us to the desired outcome because the new `index' would be identically equal to $0$ as seen from equation (\ref{lemma.sameint}). Remarkably, there is an easy way to linearize the functional $\Delta \mapsto \int_0^1 |\Delta (t)|dt$. This is achieved by dropping the absolute values and, very importantly, weighting $dt$ with the function $1-t$. These two steps lead us to the functional $\Delta \mapsto \int_0^1 \Delta (t)(1-t)dt$ and thus, in turn, to the quantity
\begin{equation}\label{equiv-2}
\mathcal{L}_{h} = \int_0^1 \left(h(t)-I_{h}(t)\right)(1-t)dt,
\end{equation}
but before declaring it an index of non-decreasingness, we need to verify that  $\mathcal{L}_{h}$ is always non-negative and takes on the value $0$ if and only if the function $h$ is non-decreasing. These are non-trivial tasks, whose solutions make up our next Theorem \ref{theorem.j}. Before formulating the theorem, we next present an illustrative example where $\mathcal{I}_{h}$ and $\mathcal{L}_{h}$ are calculated and compared.

\begin{example}\label{ex-1}\rm
For a fixed parameter $\alpha \in [0,1]$, let $h_{\alpha }$ be the function on $[0,1]$ defined by
\[
h_{\alpha}(t) =
\left
\lbrace
\begin{array}{cl}
t, & \text{for}\quad t \in \left[0,1/2\right],
\\
\alpha t + (1-\alpha)(1-t), & \text{for}\quad t \in \left(1/2,1\right].
\end{array}
\right.
\]
Note that $h_{\alpha }$ is non-decreasing when $\alpha \in [1/2,1]$, and thus $\mathcal{I}_{h_{\alpha }}=0$ and $\mathcal{L}_{h_{\alpha }}=0$. When $\alpha \in [0, 1/2)$, then a somewhat tedious calculation (relegated to Appendix \ref{appendix}) gives us the formulas
\[
\mathcal{I}_{h_{\alpha }} =
\frac{(1-2\alpha)(1-\alpha)}{2(3-2\alpha)}
\]
and
\[
\mathcal{L}_{h_{\alpha }} = \frac{(1-2\alpha)(1-\alpha)}{24}.
\]
These indices as functions of $\alpha $ are depicted in Figure \ref{fig:ihjh},
\begin{figure}[h!]
\centering
\includegraphics[width=2.5in]{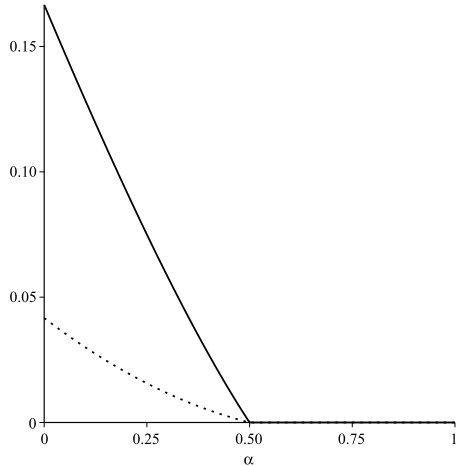}
\caption{The indices $\mathcal{I}_{h}$ (solid) and $\mathcal{L}_{h}$ (dotted) as functions of $\alpha$.}
\label{fig:ihjh}
\end{figure}
which concludes Example \ref{ex-1}.
\end{example}

\begin{theorem}\label{theorem.j}
For every function $h: [0,1] \to \mathbf{R}$, the index $\mathcal{L}_{h}$ is non-negative and takes on the value $0$ if and only if the function $h$ is non-decreasing.
\end{theorem}

\begin{proof}
The proof is somewhat complex, and we have thus subdivided it into three parts:
First, we establish an alternative representation (equation (\ref{j-equation}) below) for $\mathcal{L}_{h}$ on which the rest of the proof relies, and which, incidentally, clarifies how we came up with the weight $1-t$ in definition (\ref{equiv-2}). Then, in the second part, which is the longest and most complex part of the proof of Theorem \ref{theorem.j}, we establish a certain ordering result (bound (\ref{step-two}) below) that implies the non-negativity of $\mathcal{L}_{h}$. Finally, in the third part we prove that $\mathcal{L}_{h}=0$ if and only if the function $h$ is non-decreasing.

\paragraph{\it Part 1:}

Here we express $\mathcal{L}_{h}$ by an alternative formula that plays a pivotal role in our subsequent considerations. For this, we first recall that, by  definition, the indicator  $\textbf{1}\{ S \}$ of statement $S$ takes on the value $1$ if the statement $S$ is true and on the value $0$ otherwise. With this notation, and also using Fubini's theorem, we have the equations:
\begin{align}
\mathcal{L}_{h}
&= \int_0^1 (h(t)-I_h(t))\int_t^1 dsdt \notag \\
&= \int_0^1 \int_0^1 (h(t)-I_h(t)) \textbf{1}\{s \geq t\} dsdt \notag \\
&= \int_0^1 \left(\int_0^s h(t)dt - \int_0^s I_h(t)dt \right)ds \notag \\
&= \int_0^1 \left(H_h(s)-C_H(s)\right)ds,
\label{j-equation}
\end{align}
where $H_{h}:[0,1]\to \mathbf{R} $ is defined by
$H_h(s)= \int_0^s h(t)dt$, and $C_{H_{h}}:[0,1]\to \mathbf{R} $ is the convex rearrangement of $H_{h}$ defined by $C_H(s)=\int_0^s I_h(t)dt$, where $I_{h}$ is the non-decreasing rearrangement of $h$. The right-hand side of equation (\ref{j-equation}) is the desired alternative expression of $\mathcal{L}_{h}$.

\paragraph{\it Part 2:}

In view of expression (\ref{j-equation}), the non-negativity of $\mathcal{L}_{h}$ follows from the bound
\begin{equation}\label{step-two}
H_{h}(t) \geq C_{H_{h}}(t) \quad \textrm{for all} \quad t \in [0,1].
\end{equation}
To prove bound (\ref{step-two}), we first note that every real number $y\in \mathbf{R}$ can be decomposed as the sum $w_1(y)+w_2(y)$, where $w_1(y)=\min \{ y,0 \}$ and $w_2(y)=\max \{ y,0 \}$.
Hence,
\begin{equation}
I_{h}(s)= F_X^{-1}(s)=w_1\big ( F_X^{-1}(s) \big )+w_2\big ( F_X^{-1}(s) \big ).
\label{qq-1}
\end{equation}
Now we recall (cf., e.g., Denuit et al, 2005, Property 1.5.16(i), p.~19) that for every non-decreasing and continuous function $w$, we have the equation $w( F_X^{-1}(s) )=F_{w(X)}^{-1}(s)$. Since $w_1$ and $w_2$ are non-decreasing and continuous, equation (\ref{qq-1}) implies
\[
I_{h}(s)= F_{w_1(X)}^{-1}(s) +F_{w_2(X)}^{-1}(s) \big )
= I_{h_{-}}(s)+ I_{h_{+}}(s),
\]
where $h_{-}(s)=w_1(h(s))$ and  $h_{+}(s)=w_2(h(s))$. Hence,
\begin{align*}
\int_0^t I_{h}(s)ds
&= \int_0^t  I_{h_{-}}(s)ds+ \int_0^t I_{h_{+}}(s)ds
\notag
\\
&\le  \int_0^t  h_{-}(s)ds+ \int_0^t h_{+}(s)ds = \int_0^t h(s)ds,
\end{align*}
provided that
\begin{equation}
\int_0^t  I_{h_{-}}(s)ds \le  \int_0^t  h_{-}(s)ds
\label{qq-4a}
\end{equation}
and
\begin{equation}
\int_0^t  I_{h_{+}}(s)ds \le  \int_0^t  h_{+}(s)ds.
\label{qq-4b}
\end{equation}
We shall prove bounds (\ref{qq-4a}) and (\ref{qq-4b}) next.

\paragraph{\it Proof of bound (\ref{qq-4a}).}

Let $X_{-}=\min \{X,0\} $. We have the equation $\int_0^t  I_{h_{-}}(s)ds =\int_0^t F_{X_{-}}^{-1}(s)ds$ and thus the bound
\begin{equation}
\int_0^t  I_{h_{-}}(s)ds
\le \int_0^t F_{YX_{-}}^{-1}(s)ds,
\label{qq-4a1}
\end{equation}
where $Y$ is the random variable defined by $Y(\omega)=\mathbf{1}\{\omega \le t\}$. To establish bound (\ref{qq-4a1}), we have used the inequality $X_{-} \le YX_{-}$, which holds because $X_{-} $ is non-positive.

Next we observe that the cdf $F_{YX_{-}}(x)$ takes on the value $1$ at the point $x=0$ and has a jump of a size at least as high as $1-t$ at the point $x=0$. Hence, the quantile function $F_{YX_{-}}^{-1}(s)$ is equal to $0$ for at least all $s\in (t,1)$, and so we have the equations:
\begin{equation}
\int_0^t F_{YX_{-}}^{-1}(s)ds= \int_0^1 F_{YX_{-}}^{-1}(s)ds
= \mathbf{E}[YX_{-}]=\int_0^t  h_{-}(s)ds.
\label{qq-4a2}
\end{equation}
Bound (\ref{qq-4a1}) and equations (\ref{qq-4a2}) complete the proof of bound (\ref{qq-4a}).

\paragraph{\it Proof of bound (\ref{qq-4b}).}

Let $X_{+}=\max \{X,0\} $. In our following considerations we shall need to estimate $X_{+}$ from below by $ZX_{+}$, where $Z$ is the random variable defined by $Z(\omega)=\mathbf{1}\{\omega > t\}$. For this reason, we now observe that bound (\ref{qq-4b}) is equivalent to the following one:
\begin{equation}
\int_t^1  I_{h_{+}}(s)ds \ge  \int_t^1  h_{+}(s)ds.
\label{qq-4b1}
\end{equation}
(The equivalence of the two bounds follows from the equation $\int_0^1  I_{h_{+}}(s)ds =  \int_0^1  h_{+}(s)ds$, which is a consequence of equation (\ref{lemma.sameint}).) To establish bound (\ref{qq-4b1}), we start with the equation $\int_t^1  I_{h_{+}}(s)ds =\int_t^1 F_{X_{+}}^{-1}(s)ds$ and arrive at the bound
\begin{equation}
\int_t^1  I_{h_{+}}(s)ds
\ge \int_t^1 F_{ZX_{+}}^{-1}(s)ds.
\label{qq-4b2}
\end{equation}
The cdf $F_{ZX_{+}}(x)$ is equal to $0$ for all $x<0$ and has a jump of a size at least as high as $t$ at the point $x=0$. Hence, the quantile function $F_{ZX_{+}}^{-1}(s)$ is equal to $0$ for at least all $s\in (0,t)$, and so we have the equations:
\begin{equation}
\int_t^1 F_{ZX_{+}}^{-1}(s)ds= \int_0^1 F_{ZX_{+}}^{-1}(s)ds
= \mathbf{E}[ZX_{+}]=\int_t^1  h_{+}(s)ds.
\label{qq-4b3}
\end{equation}
Bound (\ref{qq-4b2}) and equations (\ref{qq-4b3}) complete the proof of bound (\ref{qq-4b1}) and thus, in turn, establish bound (\ref{qq-4b}) as well.

Having thus proved bounds (\ref{qq-4a}) and (\ref{qq-4b}), we have established bound (\ref{step-two}). As we have noted earlier, this implies that $\mathcal{L}_{h}$ is non-negative.

\paragraph{\it Part 3:}

In this final part of the proof of Theorem \ref{theorem.j}, we establish the fact that $\mathcal{L}_{h}$ takes on the value $0$ if and only if the function $h$ is non-decreasing. This we do in two parts.

First, we assume that $h$ is non-decreasing. Then the function $H_{h}$ is convex. Furthermore, the convex rearrangement $C_{H_{h}}$ of the function $H_{h}$ leaves the function $H_{h}$ unchanged because $H_{h}$ is convex. In summary, when $h$ is non-decreasing, then the integral $\int_0^1 \left(H_{h}(t)-C_{H_{h}}(t)\right)dt$ and thus the index $\mathcal{L}_{h}$ are equal to $0$.

Moving now in the opposite direction, if the integral $\int_0^1 \left(H_{h}(t)-C_{H_{h}}(t)\right)dt$ is equal to $0$, then due to the already proved bound $H_{h}\geq C_{H_{h}}$, we have $H_{h}(t) = C_{H_{h}}(t)$ for $\lambda $-almost all $t \in [0,1]$. Consequently, the function $H_{h}$ must be convex, and thus the function $h$ must be non-decreasing.
This concludes the proof of Step 3, and thus of the entire Theorem \ref{theorem.j}.
\end{proof}

As we have seen in the proof of Theorem \ref{theorem.j}, the definition of the index $\mathcal{L}_{h}$ fundamentally relies on the notion of convex rearrangement, which also prominently features in several other research areas, such as:
\begin{itemize}
  \item Stochastic processes (cf., e.g., Zhukova, 1994; Davydov, 1998; Thilly, 1999; Davydov and Thilly, 2002; Davydov and Zitikis, 2004; Davydov and Thilly, 2007).
  \item Convex analysis (cf., e.g., Davydov and Vershik, 1998) with applications in areas such as the optimal transport problem (cf., e.g., Lachi\'{e}ze-Rey and Davydov, 2011).
  \item Econometrics (cf., e.g., Lorenz, 1905; Gastwirth, 1971; Giorgi, 2005).
  \item Insurance (cf., e.g., Brazauskas et al, 2008; Greselin et al, 2009; Necir et al, 2010).
\end{itemize}

We conclude this section with a few properties of $\mathcal{L}_{h}$. First, when $h$ and $g$ are comonotonic, then
\begin{equation}\label{prop-1r}
\mathcal{L}_{h+g}= \mathcal{L}_{h}+\mathcal{L}_{g},
\end{equation}
which follows from equation (\ref{prop-1v2}) and the linearity of the functional $\Delta \mapsto \int_0^1 \Delta (t)(1-t)dt$. In particular, we have $\mathcal{L}_{h+d}=\mathcal{L}_{h}$ for every function $h$ and every constant $d\in \mathbf{R}$, because $\mathcal{L}_{d}=0$. Next, for every non-negative constant $c\ge 0$, we have the equation
\begin{equation}\label{prop-2aa}
 \mathcal{L}_{ch}=c\mathcal{L}_{h},
\end{equation}
which follows immediately from $I_{ch}(t)=cI_{h}(t) $ and the definition of $\mathcal{L}_{h}$. Furthermore, from the definitions of $\mathcal{I}_{h}$ and $\mathcal{L}_{h}$ we immediately obtain the bound
\begin{equation}\label{bound-0j1}
\mathcal{L}_{h}\le \mathcal{I}_{h},
\end{equation}
which, incidentally, explains the ordering of the two curves in Figure \ref{fig:ihjh}.

\section{Computing the indices}
\label{section-computing}

Except for very simple functions such as $h_{\alpha }$ of Example \ref{ex-1}, calculating the indices $\mathcal{I}_{h}$ and $\mathcal{L}_{h}$ is usually a tedious and time consuming task. To facilitate a practical implementation irrespectively of the function $h$, we next develop a technique that gives numerical values of the two indices at any prescribed precision and in virtually no time.

\subsection{General considerations}

We start with a general observation: Given two integrable functions $h,g: [0,1] \to \mathbf{R}$, we have the bound
\begin{equation}
\int_0^1 \left|I_{h}(t)-I_g(t)\right| dt
\leq \int_0^1 \left|h(t)-g(t)\right| dt,
\label{eq.2.1}
\end{equation}
which is well known (e.g., Lorentz, 1953) and has been utilized by many researchers (cf., e.g., Egorov, 1990; Zhukova, 1994; Thilly, 1999; Chernozhukov et al, 2009). In Appendix \ref{appendix} we shall give a very simple proof of bound (\ref{eq.2.1}) which will further illuminate the usefulness of the probabilistic interpretation. Due to bound (\ref{eq.2.1}), we obviously have
\begin{equation}\label{bound-0i}
\left|\mathcal{I}_{h}-\mathcal{I}_{g}\right| \leq 2\int_0^1|h(t)-g(t)|dt,
\end{equation}
Likewise, we obtain the bound
\begin{equation}\label{bound-0j}
\left|\mathcal{L}_{h}-\mathcal{L}_{g}\right| \leq 2\int_0^1\left|h(t)-g(t)\right|dt,
\end{equation}
which holds for every pair of integrable functions $g,h: [0,1] \to \mathbf{R}$. Just like bound (\ref{bound-0i}), bound (\ref{bound-0j}) helps us to develope a discretization technique for calculating the index $\mathcal{L}_{h}$ numerically. More details on the technique follow next.

Namely, we shall replace $g$ by a specially constructed estimator $\widehat{h}$ of $h$ such that the $L_1$-distance $\int_0^1 |h(t)-\widehat{h}(t)| dt$ can be made as small as desired by choosing a sufficiently small `tuning' parameter $n$. To this end, we proceed as follows. First, we partition the interval $[0,1)$ into $n$ subintervals $[(i-1)/n,i/n)$ and then choose any point $t_i $ in each subinterval. Denote $\tau_i=h(t_i)$ and let
\[
\widehat{h}(t)=
\begin{cases}
\tau_i, & \text{when}\quad  t \in \left[(i-1)/n, i/n\right),
\\
\tau_n, &\text{when}\quad t = 1 .
\end{cases}
\]
With $\tau_{1:n}\le \cdots \le \tau_{n:n}$ denoting the ordered values $\tau_1, \dots, \tau_n$, the function $G_{\widehat{h}}(x)=\lambda \{ t \in [0,1]: \widehat{h}(t) \leq x \}$ can be written as
\[
G_{\widehat{h}}(x)=
\begin{cases}
0 & \text{for  } x < \tau_{1:n},
\\
i/n & \text{for  } x \in \left[\tau_{i:n},\tau_{i+1:n}\right) , \quad 1\le i \le n-1,
\\
1 & \text{for  } x \geq \tau_{n:n} .
\end{cases}
\]
Hence, the non-decreasing rearrangement $I_{\widehat{h}}(t)= \inf \{ x \in \mathbf{R}: G_{\widehat{h}}(x) \geq t \}$ can be expressed in a computationally convenient way as
\[
I_{\widehat{h}}(t)=\tau_{i:n} \quad \textrm{for every} \quad t\in ((i-1)/n,i/n],
\]
which holds for every $i=1,\dots , n$. This implies
\begin{equation}
\mathcal{I}_{\widehat{h}}=\int_0^1 |\widehat{h}(t)-I_{\widehat{h}}(t)|dt = \frac{1}{n} \sum_{i=1}^n \left|\tau_{i:n}-\tau_i\right|.
\label{ihdiscrete}
\end{equation}
Likewise, to calculate $\mathcal{L}_{\widehat{h}}$, we use formula (\ref{equiv-2}) with $\widehat{h}$ instead of $h$, and then employ the above expressions for $\widehat{h}$ and $I_{\widehat{h}}$. We obtain
\begin{equation} \label{jhdiscrete}
\mathcal{L}_{\widehat{h}}=\int_0^1 \big ( \widehat{h}(t)-I_{\widehat{h}}(t)\big ) (1-t)dt = \frac{1}{n^2} \sum_{i=1}^{n} i \left(\tau_{i:n}-\tau_i\right).
\end{equation}
From bounds (\ref{bound-0i}) and (\ref{bound-0j}), we conclude that $|\mathcal{I}_{\widehat{h}} -\mathcal{I}_{h}|$  and $|\mathcal{L}_{\widehat{h}} -\mathcal{L}_{h}|$ do not exceed $2\int_0^1 |\widehat{h}(t)-h(t) |dt$, which converges to $0$ when $n\to \infty $ irrespectively of the chosen $t_i$'s because the function $h$ is integrable on $[0,1]$. Hence, instead of calculating the usually unwieldy  $\mathcal{I}_{h}$ and $\mathcal{L}_{h}$, we can employ formulas (\ref{ihdiscrete}) and (\ref{jhdiscrete}) and easily calculate $\mathcal{I}_{\widehat{h}} $ and $\mathcal{L}_{\widehat{h}}$ instead. Choosing a sufficiently large $n$, we can reach any desired level of accuracy. An illustration of this procedure follows next.

\subsection{An illustration with insights into the indices}

Here we calculate and interpret the indices in the case of the functions $h_1(t)=\sin(tM)$ and $h_2(t)=\cos(tM)$ defined on the interval $[0,1]$, for several values of $M$. The functions are of course simple, but we have nevertheless visualized them in Figure \ref{fig.viah}
\begin{figure}[h!]
\centering
\subfigure[$M=\pi/2$]{\includegraphics[height=2in, width=2.7in]{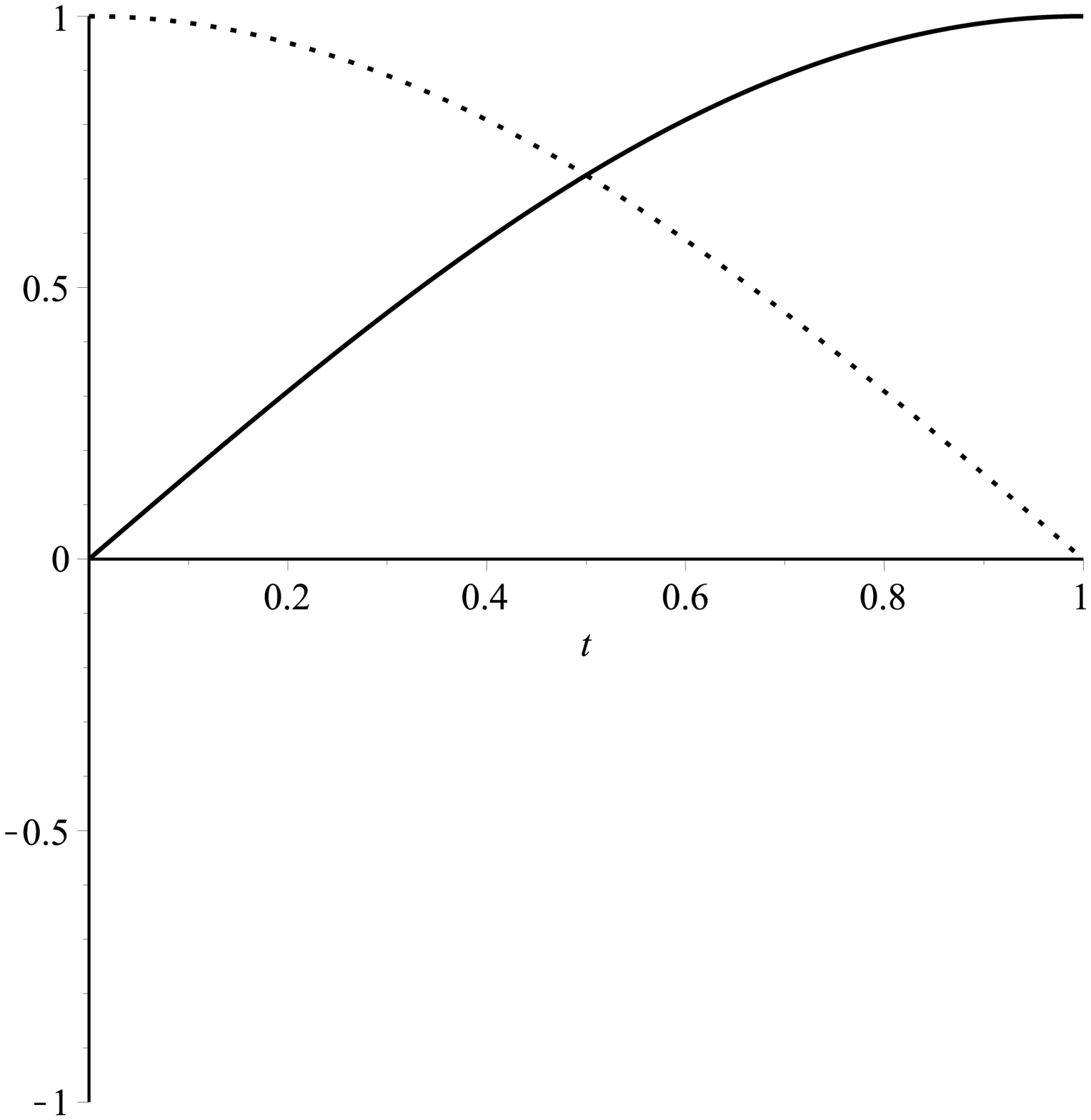}}
\qquad
\subfigure[$M=\pi$]{\includegraphics[height=2in, width=2.7in]{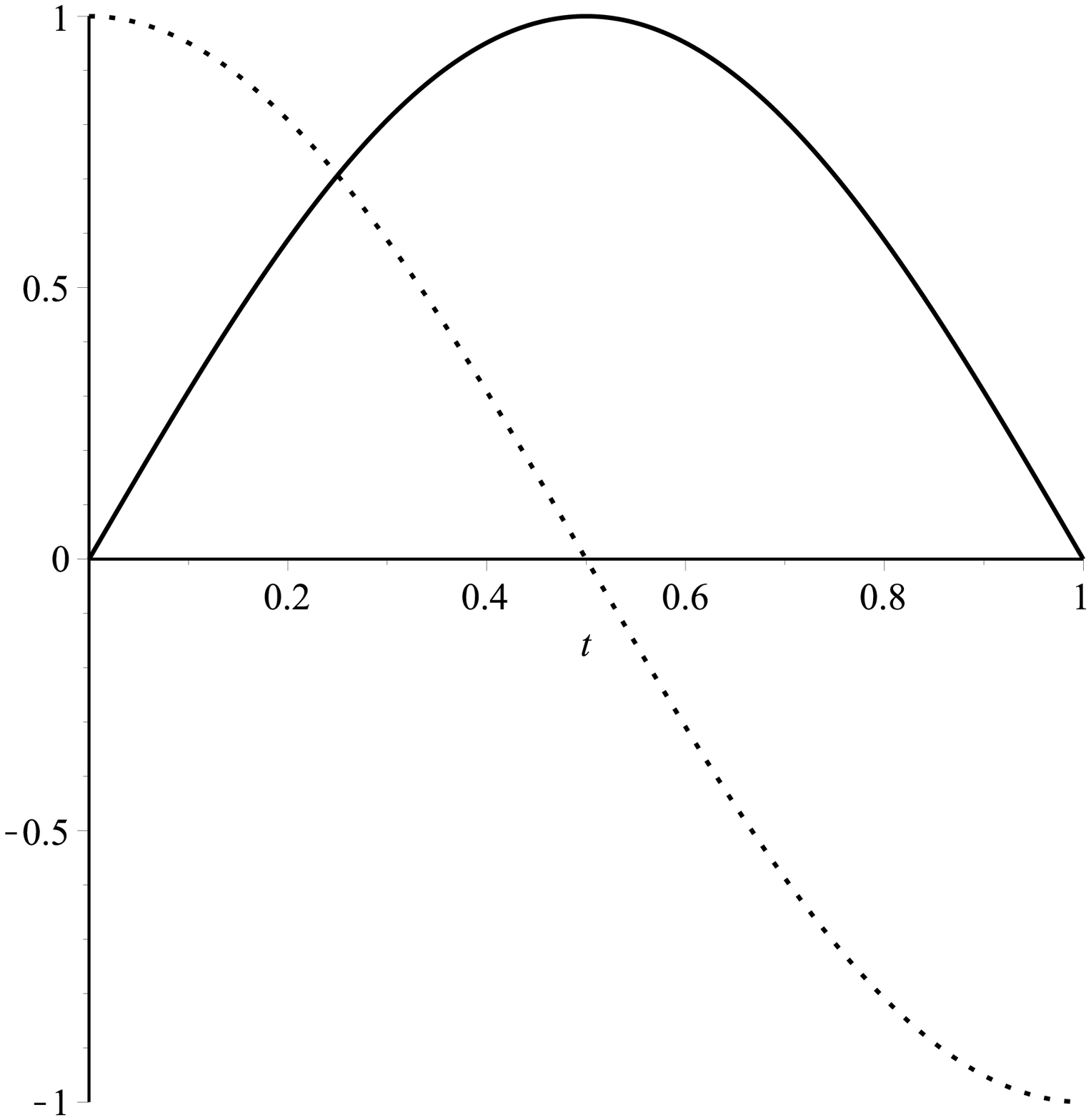}}
\centering
\subfigure[$M=3\pi/2$]{\includegraphics[height=2in, width=2.7in]{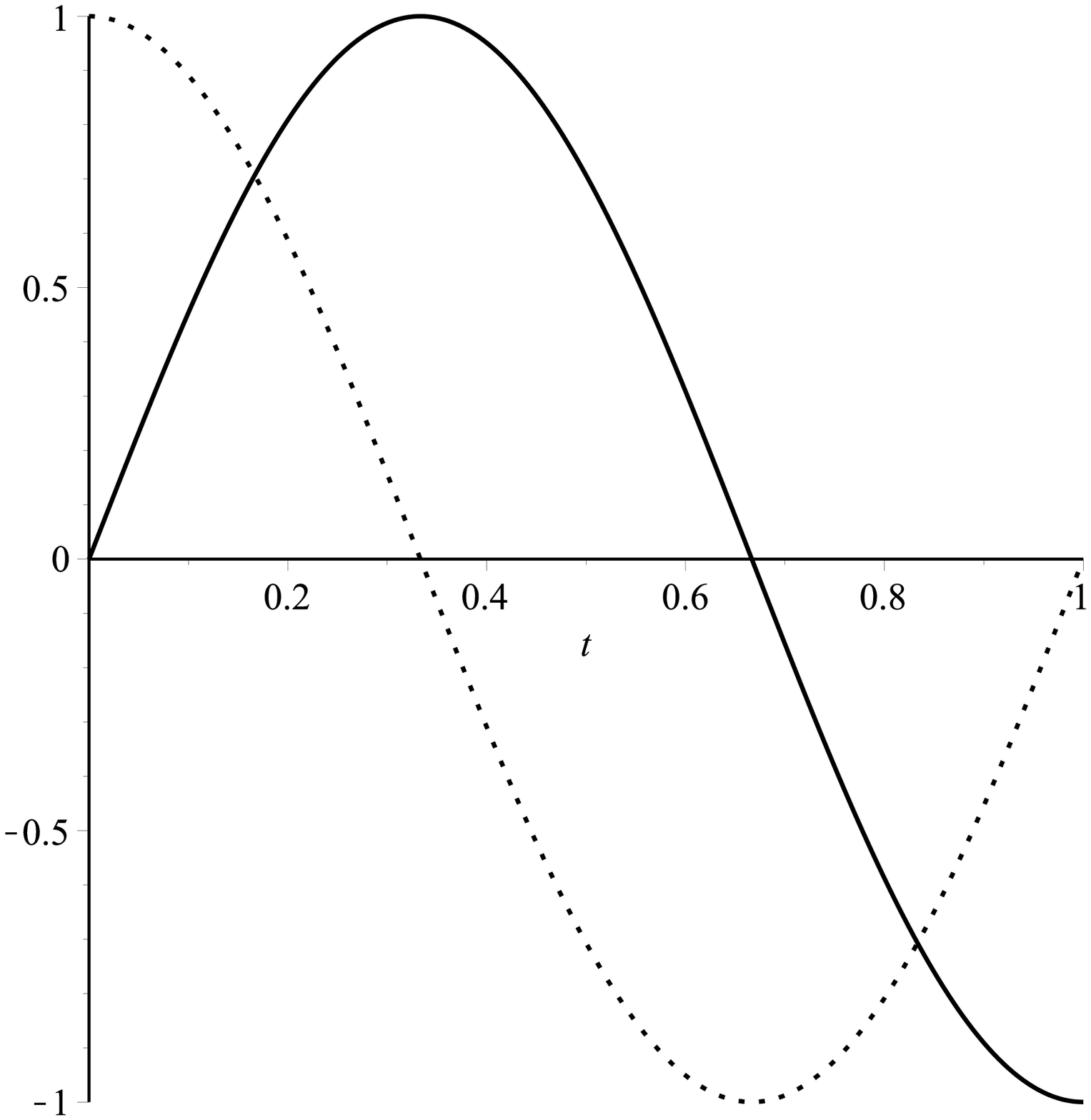}}
\qquad
\subfigure[$M=2\pi $]{\includegraphics[height=2in, width=2.7in]{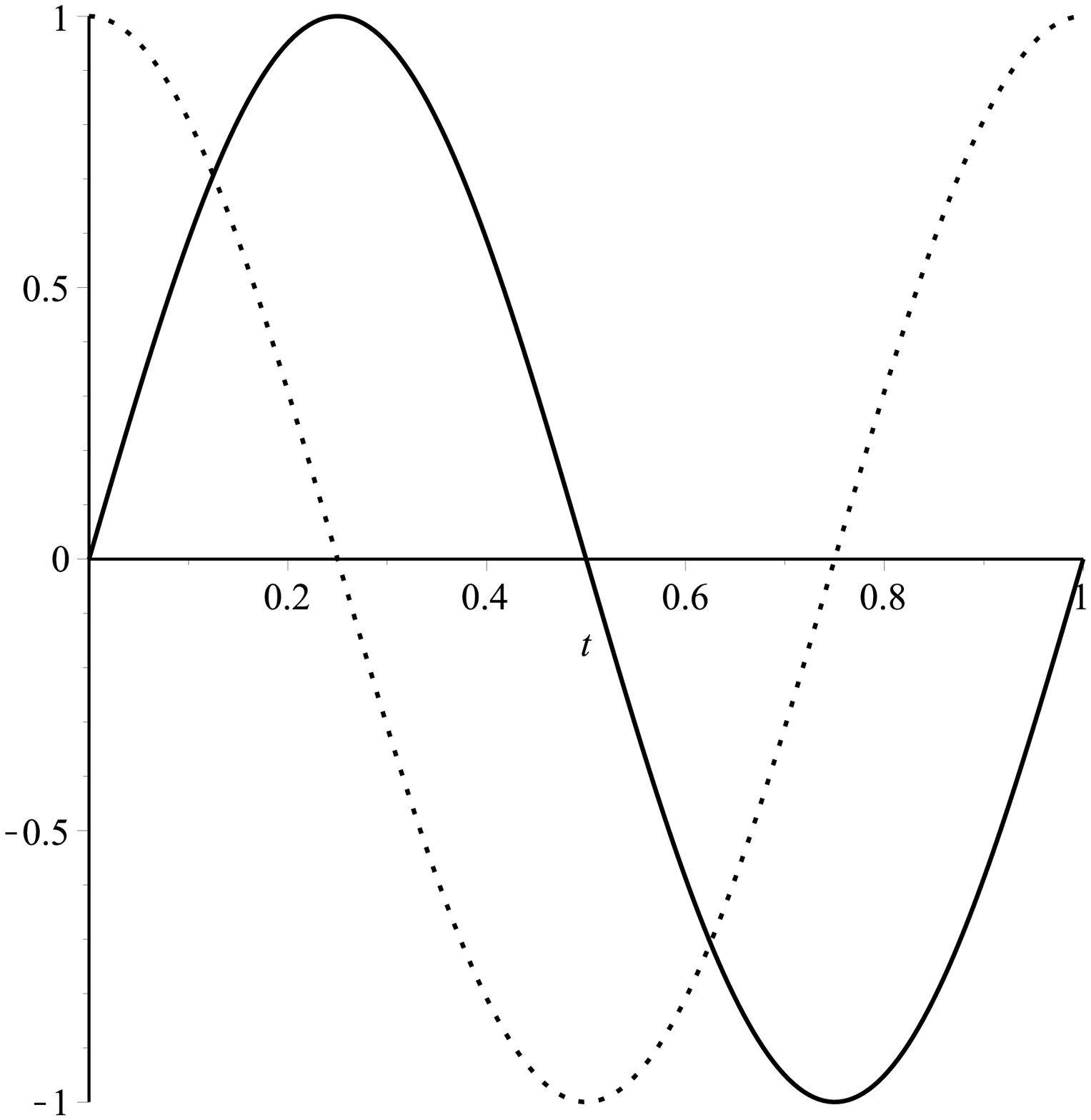}}
\caption{Functions $h_1(t)=\sin (tM)$ (solid) and $h_2(t)=\cos (tM)$ (dotted).}
\label{fig.viah}
\end{figure}
in order to facilitate our following discussion. We have used estimators (\ref{ihdiscrete}) and (\ref{jhdiscrete}) to calculate the indices, with the obtained values reported in Table \ref{tab.viah}.
\begin{table}[h!]
\centering
\begin{tabular}{cccccc}
\hline
$M$ & $\mathcal{I}_{h_1}$ & $\mathcal{I}_{h_2}$ \\
\hline
$\pi/2 $  & 0.0000 & 0.5274 \\
$\pi $    & 0.3183 & 1.2732 \\
$3\pi/2 $  & 1.1027 & 1.1027 \\
$2\pi $   & 1.2732 & 0.8270  \\
\hline
\end{tabular}
\qquad \qquad
\begin{tabular}{cccccc}
\hline
$M$ & $\mathcal{L}_{h_1}$ & $\mathcal{L}_{h_2}$ \\
\hline
$\pi/2 $  & 0.0000 &  0.1739 \\
$\pi $    & 0.0870 &  0.4053 \\
$3\pi/2 $  & 0.3409 &  0.3409 \\
$2\pi $   & 0.3618 &  0.2026 \\
\hline
\end{tabular}
\caption{Indices of non-decreasingness of $h_1(t)=\sin (tM)$ and $h_2(t)=\cos (tM)$.}
\label{tab.viah}
\end{table}
We see from the table that when $M=\pi/2$ and $\pi$, then irrespectively of which of the two indices we use, the function $h_1(t)=\sin (tM) $ is more non-decreasing (i.e., the index value is smaller) than $h_2(t)=\cos (tM) $. The two functions are equally non-decreasing when $M=3\pi/2$. When $M=2\pi $, then the function $h_1(t)=\sin (tM) $ is less non-decreasing (i.e., the index value is larger) than $h_1(t)=\cos (tM) $, and this is so for both indices. We shall now make sense of the numerical values by analyzing the four panels of Figure \ref{fig.viah}.

Panel (a) is clear: the increasing function $h_1(t)=\sin (tM) $ has its index zero, and the decreasing function $h_2(t)=\cos (tM) $ has a positive index.

In panel (b), the function $h_1(t)=\sin (tM) $ is increasing in the first half of the interval $[0,1]$ and the function $h_2(t)=\cos (tM) $ is always decreasing. Not surprisingly, therefore, any of the two indices of the function $h_1(t)=\sin (tM) $ is smaller than the corresponding index of $h_2(t)=\cos (tM) $.

In panel (c), the two functions have the same $\mathcal{I}$-indices, as well as the same $\mathcal{L}$-indices, and the reason for this is based on the general property that if $g(t)=-h(1-t)$ for all $t\in [0,1]$, then $I_g(t)=-I_h(1-t)$ for all $t\in [0,1]$. Hence, the equations $\mathcal{I}_{g}=\mathcal{I}_{h}$ and $\mathcal{L}_{g}=\mathcal{L}_{h}$ hold. In words, if we flip $h$ upside-down and also from left to right, then the value of any of the two indices will not change. This is why the two functions in panel (c) have the same $\mathcal{I}$-indices as well as the same $\mathcal{L}$-indices.

The results corresponding to panel (d) are more challenging to explain. To proceed, we adopt the following route: We subdivide the interval $(0,1]$ into four equal subintervals as follows:
\begin{equation}\label{subdivide01}
[0,1)
=\bigcup_{k=1}^{2M/\pi }\bigg [{k-1\over 2M/\pi },{k\over 2M/\pi }\bigg ) ;
\end{equation}
recall that $M=2\pi $ in this case.
By reshuffling these four subintervals, we can reconstruct the function $h_2(t)=\cos (tM) $ out of the corresponding pieces of the function $h_1(t)=\sin (tM) $, and we can of course do so the other way around. This one-to-one mapping between the two functions may wrongly suggest that the indices of the two functions should be the same, but they are obviously not, as we see from Table \ref{tab.viah}. With some tinkering we realize, however, that this is so because the original order of the aforementioned pieces of the function $h_1(t)=\sin (tM) $ is such that this function is more `wiggly' (i.e., follows the pattern  `increase-decrease-increase') than the function $h_2(t)=\cos (tM) $ (i.e., follows the pattern `decrease-increase'). Naturally now, since more wiggly functions tend to be less monotonic, the function $h_1(t)$ has a larger index than the function $h_2(t)$. Table \ref{wiggly} summarizes this point of view for all of the four panels of Figure \ref{fig.viah}.
\begin{table}[h!]
  \centering
\begin{tabular}{c|c|c|c|c}
\hline  & Panel (a) & Panel (b) & Panel (c) & Panel (d)  \\
\hline\hline
$h_1(t)$ & $+$   & $+-$ & $+--  \,~[= +-]$ & $+--+ \,~[= +-+]$\\
$\mathcal{I}_{h_1}$ & 0 & 0.3183 & 1.1027 & 1.2732 \\
\hline
$h_2(t)$ & $-$  & $-+$ & $--+  \,~[= -+]$ & $--++  \,~[= -+]$ \\
$\mathcal{I}_{h_2}$ & 0.5274 & 1.2732 & 1.1027 & 0.8270 \\
\hline
\end{tabular}
  \caption{Increasing ($-$) and decreasing ($-$) patterns of the functions $h_1(t)=\sin (tM) $ and $h_2(t)=\cos (tM) $ on the subintervals defined by equation (\ref{subdivide01}).}
  \label{wiggly}
\end{table}

\section{Indices of functions on arbitrary finite intervals}
\label{section-4}

Suppose now that we want to measure the lack of non-decreasingness of a function defined on $[a,A]\subset \mathbf{R}$. Since shifting to the left or to the right does not change the shape of the function, and thus its degree of non-decreasingness, we thus redefine the function onto the interval $[0,M]$ by simply replacing its argument $t$ by $t-a$, where $M=A-a$. Therefore, without loss of generality, from now on we work with any integrable function $f$ defined on the interval $[0,M]$, for some $M>0$. We note at the outset that we cannot reduce our task to the interval $[0,1]$ by simply replacing its argument $t$ by $tM$ because such an operation would inevitably distort the degree of non-decreasingness.

Hence, given a function $f: [0,M] \to \mathbf{R}$, we proceed by first defining its non-decreasing rearrangement by the formula
\[
I_{f,M}(t)=\inf \{ x \in \mathbf{R}: G_{f,M}(x) \geq t \}
 \quad \textrm{for all} \quad t\in [0,M],
\]
where
\[
G_{f,M}(x)= \lambda \{ t \in [0,M]: f(t) \leq x \} \quad \textrm{for all} \quad x\in \mathbf{R}.
\]
Our first index of non-decreasingness of the function $f: [0,M] \to \mathbf{R}$ is then defined by
\begin{equation}
\mathcal{I}_{f,M} = \int_0^M \left|f(t)-I_{f,M}(t)\right|dt.
\label{if}
\end{equation}
Furthermore, with
$H_{f,M}(t)= \int_0^t f(s)ds$ and $ C_{H_{f,M}}(t)=\int_0^t I_{f,M}(s)ds $ for all $ t\in [0,M]$, we define the second index of non-decreasingness of $f$ by the formula
\begin{equation}
\mathcal{L}_{f,M} = \int_0^M \left(H_{f,M}(t)-C_{H_{f,M}}(t)\right)dt.
\label{jf}
\end{equation}

We shall next illustrate the two indices using the functions $\sin(t)$ and $\cos(t)$ defined on the four domains $[0,\pi/2]$, $[0,\pi]$, $[0,3\pi/2]$, and $[0,2\pi]$. The values of the two indices are given in Table \ref{tab.direct}.
\begin{table}[h!]
\centering
\begin{tabular}{ccccc}
\hline
$M$ & $\mathcal{I}_{\sin ,M}$ & $\mathcal{I}_{\cos ,M}$ \\
\hline
$\pi/2$   & 0.0000 & 0.8284  \\
$\pi $    & 1.0000 & 4.0000  \\
$3\pi/2 $ & 5.1962 & 5.1962  \\
$2\pi $   & 8.0000 & 5.1962  \\
\hline
\end{tabular}
\qquad \qquad
\begin{tabular}{ccccc}
\hline
$M$ & $\mathcal{L}_{\sin ,M}$ & $\mathcal{L}_{\cos ,M}$ \\
\hline
$\pi/2$   & 0.0000 & 0.4292 \\
$\pi $    & 0.8584 & 4.0000 \\
$3\pi/2 $ & 7.5708 & 7.5708 \\
$2\pi $   & 14.2832 & 8.0000 \\
\hline
\end{tabular}
\caption{Indices of non-decreasingness of $\sin(t)$ and $\cos(t)$ on the domain $[0,M]$.}
\label{tab.direct}
\end{table}
Since this example mimics that of Section \ref{section-computing}, various interpretations there apply here as well. In short, we see from the table that irrespectively of which of the two non-decreasing indices we use, the index of non-decreasingness of $\sin (t) $ is smaller than that of $\cos (t) $ on the domains $[0, \pi/2]$ and $[0,\pi]$. The two functions have the same non-decreasingness indices on $[0,3\pi/2]$. Finally, on the domain $[0, 2\pi]$, the index of non-decreasingness of the function $\sin (t) $ is greater than that of $\cos (t) $, irrespectively of which of the two indices we use, which implies that $\sin (t) $ is less non-decreasing than $\cos (t) $ on $[0, 2\pi]$.

We have used a discretization technique to calculate the values reported in Table \ref{tab.direct}. The technique is a modification of that of Section \ref{section-computing}. To explain the modification, in Theorem \ref{th-31} below we establish a connection between  the pair of the earlier introduced indices on the interval $[0,1]$ and the pair of the current ones on the interval $[0,M]$.

\begin{theorem}\label{th-31}
Let $f: [0,M] \to \mathbf{R}$ for some $M>0$, and let $h: [0,1] \to \mathbf{R}$ be the function defined by $h(t)=f(tM)$ for all $t \in [0,1]$. Then
\begin{equation}
\mathcal{I}_{f,M} =M \mathcal{I}_{h}
\quad \textrm{and} \quad
\mathcal{L}_{f,M} =M^2 \mathcal{L}_{h}.
\label{eq.jhjf}
\end{equation}
\end{theorem}

\begin{proof}
Since $G_{h}(x)= G_{f,M}(x)/M$, we have $I_{h}(t)= I_{f,M}(tM)$ for all $t\in [0,1]$. Hence,
\[
\mathcal{I}_{h}
= \int_0^1 \left|f(tM)-I_{f,M}(tM)\right|dt
=\frac{1}{M} \mathcal{I}_{f,M},
\]
which establishes the first equation of (\ref{eq.jhjf}). To prove the second equation, we first check that $H_{h}(t)= H_{f,M}(tM)/M$
and $C_{H_h}(t)= C_{H_{f,M}}(tM)/M$. Consequently,
\[
\mathcal{L}_{h} = \frac{1}{M}\int_0^1 \left(H_{f,M}(tM)-C_{H_{f,M}}(tM)\right)dt
= \frac{1}{M^2}\mathcal{L}_{f,M}.
\]
This establishes the second equation of (\ref{eq.jhjf}), and concludes the proof of Theorem \ref{th-31}.
\end{proof}

We are now in the position to introduce estimators $\widehat{\mathcal{I}}_{f,M}$ and $\widehat{\mathcal{L}}_{f,M}$ of the indices $\mathcal{I}_{f,M}$ and $\mathcal{L}_{f,M}$, respectively. Namely, with $h(t)=f(tM)$ and using formulas (\ref{ihdiscrete}) and (\ref{jhdiscrete}), we have
\begin{equation}
\widehat{\mathcal{I}}_{f,M} = \frac{M}{n} \sum_{i=1}^n \left|\tau_{i:n}-\tau_i\right|.
\label{ihdiscrete-1}
\end{equation}
and
\begin{equation}
\widehat{\mathcal{L}}_{f,M} = \bigg (\frac{M}{n}\bigg)^2 \sum_{i=1}^{n} i \left(\tau_{i:n}-\tau_i\right),
\label{jhdiscrete-1}
\end{equation}
where $\tau_{1:n}\le \cdots \le \tau_{n:n}$ denote the ordered values $\tau_i=f(t_i M)$, $i=1,\dots , n$. We used formulas (\ref{ihdiscrete-1}) and (\ref{jhdiscrete-1}) to obtain the numerical values of the two indices reported in Table \ref{tab.direct}, where we set $n=100,000$ in order to have a mesh sufficiently fine to achieve the desired accuracy level of four decimal digits.

\section{Conclusions}
\label{section-8}

Inspired by examples from a number of research areas, in this paper we have explored two indices designed for measuring the lack of monotonicity in functions. The indices take on the value $0$ for every non-decreasing function, and on positive values for other functions: the larger the values, the less non-decreasing the function is deemed to be. One of the two indices is simpler, but it is only subadditive for comonotonic functions, whereas the other index is more complex, but it is additive for comonotonic functions. Since the two indices are too involved to easily yield values even for elementary functions, we have devised a numerical procedure for calculating the two indices in virtually no time and at any specified accuracy.

\section*{Acknowledgments}

The first author gratefully acknowledges his PhD study support by the Directorate General of Higher Education, Ministry of National Education, Indonesia. The second author has been supported by the Natural Sciences and Engineering Research Council (NSERC) of Canada.

\appendix
\section{Technicalities}
\label{appendix}

\begin{proposition}\label{equiv-0}
Function $h:[0,1]\to \mathbf{R}$ is non-decreasing if and only if the equation $I_{h}(t)=h(t)$ holds for $\lambda $-almost all $t \in [0,1]$. If $h$ is left-continuous, then it is non-decreasing if and only  $I_{h}(t)=h(t)$ for all $t \in [0,1]$.
\end{proposition}

\begin{proof}
Assume first that $I_{h}(t)=h(t)$ for $\lambda $-almost all $t \in [0,1]$. Since the function $I_{h}$ is non-decreasing, then the function $h$ must be non-decreasing as well.

Conversely, suppose that the function $h$ is non-decreasing. Then from the definition of $G_{h}(x)$, we have the equation $G_{h}(x) = \sup \{t \in [0,1]: h(t) \leq x\}$  and thus, in turn, from  the definition of $I_{h}(t)$, we have the equation $I_{h}(t) = \lim_{s \uparrow t} h(s)$. Consequently, $I_{h}$ is left-continuous  and the equation $I_{h}(t)=h(t)$ holds at every continuity point $t\in [0,1]$ of the function $h$. Since the set of all discontinuity points of every non-decreasing function can only be at most of $\lambda $-measure zero, the converse of Proposition \ref{equiv-0} follows. This finishes the entire proof of Proposition \ref{equiv-0}.
\end{proof}

\begin{proof}[Technicalities of Example \ref{ex-1}]
We only need to consider the case $\alpha \in [0,1/2) $. Since
\[
G_{h_{\alpha }}(x) =
\left
\{
\begin{array}{ll}
x & \text{when}\quad x \in [0, \alpha ],
\\
\displaystyle
\frac{2-2\alpha}{1-2\alpha}~x+\frac{\alpha}{2\alpha-1} & \text{when}\quad x \in \left(\alpha, 1/2\right],
\\
1 & \text{when}\quad x > 1/2,
\end{array}
\right.
\]
the non-decreasing rearrangement of $h_{\alpha }$ can be expressed as follows:
\[
I_{h_{\alpha }}(t) =
\left \{
\begin{array}{ll}
t & \text{when}\quad t \in [0, \alpha),
\\
\displaystyle
\frac{1-2\alpha}{2-2\alpha} ~t
+\frac{\alpha}{2-2\alpha}& \text{when}\quad t \in [\alpha, 1].
\end{array}
\right.
\]
Utilizing the easily checked fact that the functions  $h_{\alpha}$ and $I_{h_{\alpha}}$ cross at the only point $t_c= (\alpha-2)/(2\alpha-3)$, we calculate the index $\mathcal{I}_{h_{\alpha}}$ as follows:
\begin{align*}
\mathcal{I}_{h_{\alpha }}
&= \int_{\alpha}^{1/2} |h_{\alpha}(t)-I_{h}(t)|dt+ \int_{1/2}^{{t_c}} |h_{\alpha}(t)-I_{h}(t)|dt + \int_{t_c}^1 |h_{\alpha}(t)-I_{h}(t)|dt \notag \\
&= \int_{\alpha}^{1/2} \left(t - \frac{(1-2\alpha)t+\alpha}{2-2\alpha}\right)dt
+\int_{1/2}^{t_c} \left((2\alpha-1)t+(1-\alpha)-\frac{(1-2\alpha)t+\alpha}{2-2\alpha}\right)dt \notag \\
& \text{\space\space\space}
 +\int_{t_c}^1 \left(\frac{(1-2\alpha)t+\alpha}{2-2\alpha}-(2\alpha-1)t-(1-\alpha)\right)dt \notag \\
&= \left(\frac{1}{8}\frac{(2\alpha-1)^2}{2-2\alpha}\right)
+\left(\frac{1}{2}\frac{(1-2\alpha)(\alpha-2)^2}{(2-2\alpha)(3-2\alpha)}-\frac{1}{8}\frac{(1-2\alpha)(5-2\alpha)}{2-2\alpha}\right) \notag \\
&\text{\space\space\space} +\left(\frac{\alpha-1/2}{2-2\alpha}+\frac{1}{2}\frac{(1-2\alpha)(\alpha-2)^2}{(2-2\alpha)(3-2\alpha)}\right) \notag \\
&=\frac{(1-2\alpha)(1-\alpha)}{2(3-2\alpha)} .
\end{align*}
Similar arguments produce a formula for the index $\mathcal{L}_{h_{\alpha}}$:
\begin{align*}
\mathcal{L}_{h_{\alpha }}
&= \int_{\alpha}^{1/2} (h_{\alpha}(t)-I_{h}(t))(1-t)dt+ \int_{1/2}^{1} (h_{\alpha}(t)-I_{h}(t))(1-t)dt
\\
&=\frac{(1-2\alpha)(1-\alpha)}{24}.
\end{align*}
This concludes the technicalities of Example \ref{ex-1}.
\end{proof}

\begin{proof}[Proof of bound (\ref{eq.2.1})]
Using the probabilistic interpretation, we write the equation
\begin{equation}
\int_0^1 \left|I_{h}(t)-I_g(t)\right| dt = \int_0^1 \left|F_X^{-1}(t)-F_Y^{-1}(t)\right| dt.
\label{eq.2.1a}
\end{equation}
The integral on the right-hand side of equation (\ref{eq.2.1a}) is known as the Dobrushin distance between the two cdf's $F_X$ and $F_Y$.  The integral is equal (Dobrushin, 1970) to $\inf \mathbf{E}[|\xi-\eta |\,]$, where the infinum is taken over all random variables $\xi $ and $\eta $ that have finite first moments and whose cdf's are equal to $F_X$ and $F_Y$, respectively. The infinum is not larger than $\mathbf{E}[|h(U)-g(U) |\,]$, where $U$ is a uniform random variable on $\Omega =[0,1]$, because the cdf's of the random variables $h(U)$ and $g(U)$ are equal to $F_X$ and $F_Y$, respectively. Indeed, in the case of $h(U)$ for example, the cdf $F_{h(U)}$ of $h(U)$ is equal to $\mathbf{P}\{ \omega \in \Omega : h(U(\omega ))\le x\}$, which is equal to $ \lambda \{ t \in [0,1] : h(t)\le x\}$ because $U(\omega )=\omega $ by the definition of the uniform random variable on $\Omega =[0,1]$. Note that $ \lambda \{ t \in [0,1] : h(t)\le x\}$ is equal to $G_{h}(x)$, which is in turn equal to $F_X(x)$ according to our probabilistic interpretation. Hence, $F_{h(U)}=F_X$ and, likewise, $F_{g(U)}=F_Y$. This concludes the proof of bound (\ref{eq.2.1}).
\end{proof}

\end{document}